\newif\ifprocs
\procstrue
\procsfalse 

\ifprocs
\documentclass[twoside,leqno,twocolumn]{article}
\usepackage{ltexpprt}
\else
\documentclass[11pt]{article}
\usepackage[margin=1.0in]{geometry}
\usepackage{amsthm}
\fi

\usepackage[utf8]{inputenc}
\usepackage[english]{babel}
\usepackage{amssymb}
\usepackage{amsfonts}
\usepackage{amsmath}
\usepackage{verbatim} 
\usepackage{xcolor,graphicx}
\usepackage{xspace}
\usepackage{stmaryrd}

\usepackage{ifpdf}
\ifpdf    
\usepackage[pdftex,colorlinks,linkcolor=blue,citecolor=blue,filecolor=blue,urlcolor=blue]{hyperref}
\else    
\usepackage[hypertex,colorlinks,linkcolor=blue,citecolor=blue,filecolor=blue,urlcolor=blue]{hyperref}
\fi

\ifprocs
\renewcommand{\paragraph}{\subparagraph}
\newtheorem{question}[@theorem]{Question}
\newtheorem{conjecture}[@theorem]{Conjecture}
\newtheorem{remark}[@theorem]{Remark}
\newtheorem{claim}[@theorem]{Claim}

\newtheorem{lemma}[@theorem]{Lemma}

\newtheorem{corollary}[@theorem]{Corollary}

\newtheorem{fact}[@theorem]{Fact}

\else
\newtheorem{theorem}{Theorem}[section]
\newtheorem{lemma}[theorem]{Lemma}

\newtheorem{corollary}[theorem]{Corollary}

\newtheorem{fact}[theorem]{Fact}
\newtheorem{remark}[theorem]{Remark}

\fi

\newcommand{\eqdef}{:=}

\DeclareMathOperator{\supp}{supp}
\DeclareMathOperator{\dist}{str}
\DeclareMathOperator*{\E}{{\mathbb E}}

\let\Pr\relax
\DeclareMathOperator*{\Pr}{\mathbb{P}}
\newcommand{\Terminals}{\mathsf{T}}
\newcommand{\Term}{\Terminals}
\newcommand{\Trees}{\mathsf{Trees}}
\newcommand{\gap}{\mathsf{gap}}

\newcommand{\conetrms}{c_1^+}
\DeclareMathOperator{\capa}{\mathsf{cap}}
\DeclareMathOperator{\dem}{\mathsf{dem}}
\DeclareMathOperator{\mcf}{\mathrm{mcf}}
\newcommand{\f}{\varphi}
\newcommand{\id}{\mathsf{id}}
\newcommand{\diam}{\mathrm{diam}}
\newcommand{\dil}{\mathrm{dil}}

\newcommand{\N}{\mathbb N}
\newcommand{\R}{\mathbb R}

\newcommand{\cF}{\mathcal{F}}

\newcommand\norm[1]{\left\lVert#1\right\rVert_1}

\title{Flow-Cut Gaps and Face Covers in Planar Graphs}

\author{Robert Krauthgamer%
\thanks{Weizmann Institute of Science, Israel.
Work partially supported by ONR Award N00014-18-1-2364, the Israel Science Foundation grant \#1086/18, and a Minerva Foundation grant.
Email: \texttt{robert.krauthgamer@weizmann.ac.il}
}
\footnote{Part of this work was done while the author was visiting the Simons Institute for the Theory of Computing.
}
\\
Weizmann Institute
\and
James R. Lee%
\thanks{University of Washington, Seattle, WA.
   Work supported by NSF
   grants CCF-1616297 and CCF-1407779 and a Simons Investigator Award.
Email: \texttt{jrl@cs.washington.edu}
}
\footnotemark[2]
\\
University of Washington
\and
Havana (Inbal) Rika%
\thanks{Weizmann Institute of Science, Israel.
Email: \texttt{havana.rika@weizmann.ac.il}
}
\\
Weizmann Institute
}

\ifprocs
\else
\fi

\begin{document}
\date{}
\maketitle

\begin{abstract}
The relationship between the
sparsest cut and the maximum concurrent multi-flow
in graphs has been studied extensively.
For general graphs, the worst-case gap between
these two quantities is now settled:  When there are $k$ terminal pairs,
the flow-cut gap is $O(\log k)$, and this is tight.
But when topological restrictions are placed on the flow
network, the situation is far less clear.
In particular, it has been conjectured that
the flow-cut gap in planar networks is $O(1)$, while
the known bounds place the gap somewhere between $2$ (Lee and Raghavendra, 2003) and $O(\sqrt{\log k})$ (Rao, 1999).

A seminal result of Okamura and Seymour (1981) shows that when all the terminals
of a planar network lie on a single face, the flow-cut gap is exactly $1$.
This setting can be generalized by considering planar networks where
the terminals lie on one of $\gamma > 1$ faces in some fixed planar drawing.
Lee and Sidiropoulos (2009) proved that the flow-cut gap is bounded
by a function of $\gamma$, and Chekuri, Shepherd, and Weibel (2013) showed
that the gap is at most $3 \gamma$.
We significantly improve
these asymptotics by establishing that the flow-cut gap is $O(\log \gamma)$.
This is achieved by showing that the edge-weighted shortest-path metric induced
on the terminals admits a stochastic embedding into trees with distortion $O(\log \gamma)$.
The latter result is tight, e.g., for a square planar lattice on $\Theta(\gamma)$ vertices.

The preceding results refer to the setting of {\em edge-capacitated} networks.
For {\em vertex-capacitated networks,} it can be significantly more challenging
to control flow-cut gaps.
While there is no exact vertex-capacitated version of the Okamura-Seymour Theorem,
an approximate version holds;
Lee, Mendel, and Moharrami (2015) showed that
the vertex-capacitated flow-cut gap is $O(1)$ on planar networks whose
terminals lie on a single face.
We prove that the flow-cut gap is $O(\gamma)$ for vertex-capacitated
instances when the
terminals lie on at most $\gamma$ faces.
In fact, this result holds in the more
general setting of submodular vertex capacities.
\end{abstract}

\ifprocs
\else
\newpage
{\small
\begingroup
\hypersetup{linktocpage=false}
\setcounter{tocdepth}{2}
\linespread{0.5}
\tableofcontents
\endgroup
}
\fi

\section{Introduction}
We present some new upper bounds on the gap between the concurrent flow and sparsest cut in planar graphs in
terms of the topology of the terminal set.
Our proof employs low-distortion metric embeddings into $\ell_1$,
which are known to have a tight connection to the flow-cut gap (see, e.g., \cite{LLR95,GNRS04}).
We now review the relevant terminology.

Consider an undirected graph $G$ equipped with nonnegative edge lengths $\ell : E(G) \to \R_+$ and a subset
$\Terminals = \Terminals(G) \subseteq V(G)$ of {\em terminal vertices.}
 We use $d_{G,\ell}$ to denote the shortest-path distance in $G$, where the length
of paths is computed using the edge lengths $\ell$.
We use $\conetrms(G,\ell;\Terminals)$ to denote the minimal number $D \geq 1$
for which there exists $1$-Lipschitz mapping $F : V(G) \to \ell_1$ such that
$F|_{\Terminals(G)}$ has bilipschitz distortion $D$.  In other words,
\ifprocs
\else
\begin{align}
  \label{EQ: item 1 ell1 embedding}
  \forall u,v\in V(G):\qquad
  &\norm{f(u)-f(v)}\leq d_{G,\ell}(u,v)\,,
  \\
  \label{EQ: item 2 ell1 embedding}
  \forall s,t\in \Terminals(G):\qquad
  &\norm{f(s)-f(t)}\geq \tfrac{1}{D}\cdot d_{G,\ell}(s,t)\,.
\end{align}
\fi

For an undirected graph $G$,
we define $\conetrms(G; \Terminals) \eqdef \sup_{\ell} \conetrms(G,\ell; \Terminals)$,
where $\ell$ ranges over all nonnegative lengths $\ell : E(G) \to \R_+$.
When $\Terminals=V(G)$, we may omit it and write
$\conetrms(G,\ell) \eqdef \conetrms(G,\ell;V(G))$
and $\conetrms(G) \eqdef \conetrms(G;V(G))$.
Finally, for a family $\mathcal{F}$ of finite graphs,
we denote $\conetrms(\mathcal{F}) \eqdef \sup \{ \conetrms(G) : G \in \mathcal{F} \}$,
and for $k \in \N$, we denote
\[
   \conetrms(\mathcal{F}; k) \eqdef \sup \left\{ \conetrms(G;\Terminals) : G \in \mathcal{F}, \Term \subseteq V(G), |\Term|=k \right\}\,.
\]

Let $\cF_{\mathrm{fin}}$ denote the family of all finite graphs,
and $\cF_{\mathrm{plan}}$ the family of all planar graphs.
It is known that $\conetrms(\cF_{\mathrm{fin}}; k)=\Theta(\log k)$~\cite{AR98,LLR95} for all $k \geq 1$.
For planar graphs, one has
$\conetrms(\cF_{\mathrm{plan}}; k) \leq O(\sqrt{\log k})$ \cite{Rao99} and $\conetrms(\cF_{\mathrm{plan}}) \geq 2$ \cite{LR10}.

Fix a plane graph $G$ (this is a planar graph $G$ together with a drawing in the plane).
For $\Term \subseteq V(G)$, we define the quantity $\gamma(G;\Term)$ to be the
smallest number of faces in $G$ that together cover all the vertices of $\Term$,
and $\gamma(G) \eqdef \gamma(G;V(G))$.

We say that the pair $(G,\Term)$ is an \emph{Okamura-Seymour instance},
or in short an \emph{OS-instance},
if it can be drawn in the plane with all its terminal on the same face,
i.e., if there is a planar representation for which $\gamma(G;\Term)=1$.
A seminal result of Okamura and Seymour \cite{OS81}
implies that $\conetrms(G;\Term)=1$ whenever $(G,\Term)$ is an OS-instance.

The methods of~\cite{LS09}
show that $\conetrms(G; \Term)\leq 2^{O(\gamma(G;\Term))}$,
and a more direct proof of~\cite[Theorem 4.13]{CSW13}
later showed that $\conetrms(G; \Term)\leq 3\gamma(G;\Term)$.
Our main result is the following improvement.

\begin{theorem}\label{THM: main result embedding}
   For every plane graph $G$ and terminal set $\Term \subseteq V(G)$,
      \[\conetrms(G;\Term) \leq O(\log \gamma(G;\Term)).\]
\end{theorem}

A long-standing conjecture \cite{GNRS04}
asserts that $\conetrms(\cF) < \infty$ for every family $\cF$ of finite graphs
that is closed under taking minors and does not contain all finite graphs.
If true, this conjecture would of course imply that one
can replace the bound of Theorem~\ref{THM: main result embedding} with a universal constant.

It is known that a plane graph $G$ has treewidth $O(\sqrt{\gamma(G)})$ \cite{KLL02}.
If we use $\cF_{\mathrm{tw}}(w)$ and $\cF_{\mathrm{pw}}(w)$ to denote the families of graphs of treewidth $w$ and
pathwidth $w$, respectively, then
it is known that $\conetrms(\cF_{\mathrm{tw}}(2))$ is finite \cite{GNRS04},
but this remains open for $\conetrms(\cF_{\mathrm{tw}}(3))$.
(On the other hand, $\conetrms(\cF_{\mathrm{pw}}(w))$ is finite for every $w \geq 1$ \cite{LS13}, and
currently the best quantitative bound
is $\conetrms(\cF_{\mathrm{pw}}(w)) \leq O(\sqrt{w})$ \cite{AFGN18}.)

The parameter $\gamma(G;\Term)$
was previously studied in the context of other computational problems,
including the Steiner tree problem \cite{EMV87, Bern90, KNL19},
all-pairs shortest paths \cite{Fred95},
and cut sparsifiers \cite{KR17arxiv,KPZ18}.
For a planar graph $G$ (without a drawing) and $\Term \subseteq V(G)$,
the \emph{terminal face cover}, denoted $\gamma^*(G;\Term)$,
is the minimum number of faces that cover $\Term$
in all possible drawings of $G$ in the plane.
All our results, including Theorems~\ref{THM: main result embedding},
\ref{THM: main result flow-cut gap}, and~\ref{THM: UB embedding into trees},
hold also for the parameter $\gamma^*(G;\Term)$,
simply because the relevant quantities do not depend on the graph's drawing.
When $G$ and $T$ are given as input,
$\gamma(G;\Term)$ can be computed in polynomial time~\cite{BM88},
but computing $\gamma^*(G;\Term)$ is NP-hard~\cite{BM88}.
In other words, while finding faces that cover $\Term$ optimally
in a given drawing is tractable,
finding an optimal drawing is hard.

\subsection{The flow-cut gap}
\label{SEC: intro. flow cut gap}

We now define the flow-cut gap, and briefly explain its connection to $\conetrms$.
Consider an undirected graph $G$ with terminals $\Term = \Terminals(G)$.
Let $c : E(G) \to \R_+$ denote an assignment of {\em capacities} to edges,
and $d : {\binom{\Term}{2}} \to \R_+$ an assignment of {\em demands.}
The triple $(G,c,d)$ is called an (undirected) {\em network.}
The \emph{concurrent flow} value of the network is the maximum value $\lambda>0$,
such that $\lambda\cdot d(\{s,t\})$ units of flow can be routed between
every demand pair $\{s,t\} \in \binom{\Term}{2}$, simultaneously but as separate commodities,
without exceeding edge capacities.

Given the network $(G,c,d)$ and a subset $S\subset V$,
let $\capa(S)$ denote the total capacity of edges crossing the cut $(S, V\setminus S)$,
and let $\dem(S)$ denote the sum of demands $d(\{s,t\})$ over all pairs $\{s,t\} \in \binom{\Term}{2}$ that cross the same cut.
The \emph{sparsity of a cut $(S,V\setminus S)$} is defined as $\capa(S)/\dem(S)$,
and the \emph{sparsest-cut value of $(G,c,d)$} is the minimum sparsity over all cuts in $G$.
Finally, the
\emph{flow-cut gap} in the network $(G,c,d)$ is defined as the ratio
\[
  \gap(G,c,d)
  \eqdef \frac {\operatorname{sparsest-cut}(G,c,d)}
               {\operatorname{concurrent-flow}(G,c,d)}
  \geq 1\,,
\]
where the inequality is a basic exercise.

For a graph $G$ (without capacities and demands),
denote $\gap(G;\Terminals) \eqdef \sup_{c,d} \gap(G,c,d)$,
where $c$ and $d : \binom{\Term}{2} \to \R_+$ range over assignments of capacities and demands as above.
The following theorem presents the fundamental duality between flow-cut gaps and $\ell_1$ distortion.

\begin{theorem}[\cite{AR98,LLR95,GNRS04}]\label{THM: flow cut = embedding}
  For every finite graph $G$ with terminals $\Terminals\subseteq V(G)$,
  \[ \gap(G; \Terminals)=\conetrms(G; \Terminals) \, . \]
\end{theorem}

Thus our main result (Theorem~\ref{THM: main result embedding})
can be stated in terms of flow-cut gaps as follows.

\begin{theorem}\label{THM: main result flow-cut gap}
   For every plane graph $G$ and terminal set $\Term \subseteq V(G)$,
   \[\gap(G; \Terminals) \leq O(\log \gamma(G;\Term))\,.\]
\end{theorem}

\begin{remark}
   It is straightforward to check that our argument
   yields a polynomial-time algorithm that, given a plane graph $G$
   and capacities $c$ and demands $d : \binom{\Term}{2} \to \R_+$,
   produces a cut $(S, V(G) \setminus S)$ whose sparsity
   is within an $O(\log \gamma(G;\Term))$ factor of the sparsest cut
   in the flow network $(G,c,d)$.
\end{remark}

\subsection{The vertex-capacitated flow-cut gap}

One can consider
the analogous problems in more general networks;
for instance, those which are \emph{vertex-capacitated} (instead of edge-capacitated).
In that setting, bounding the flow-cut gap appears to be significantly more challenging than for edge capacities.
The authors of \cite{FHL05} establish that the vertex-capacitated flow-cut gap is $O(\log k)$
for general networks with $k$ terminals, and this bound is known to be tight \cite{LR99}.

For planar networks, Lee, Mendel, and Moharrami~\cite{LMM15}
sought a vertex-capacitated version of the Okamura-Seymour Theorem \cite{OS81},
and proved that the vertex-capacitated flow-cut gap is $O(1)$
for instances $(G,\Term)$ satisfying $\gamma(G;\Term)=1$.

However, it was not previously known whether the gap is bounded even for $\gamma(G;\Term)=2$.
We prove that in planar vertex-capacitated networks $(G,\Term)$
with $\gamma=\gamma(G;\Term)$, the flow-cut gap is $O(\gamma)$; see Theorem~\ref{thm:poly-flowcut}.
In fact, we prove this result in the more general setting of submodular vertex capacities, also known as \emph{polymatroid networks.}
This model was introduced in \cite{CKRV15} as a generalization of vertex capacities,
and the papers \cite{CKRV15,LMM15} showed that more refined methods in metric embedding
theory are able to establish upper bounds on the flow-cut gap
even in this general setting.

\subsection{Stochastic embeddings}
\label{SEC: intro. embedding}

Instead of embedding plane graphs with a given $\gamma(G;\Term)$ directly into $\ell_1$,
we will establish the stronger result that such instances can be randomly approximated
by trees in a suitable sense.

If $(X,d_X)$ is a finite metric space and $\cF$ is a family of finite metric spaces, then a
{\em stochastic embedding of $(X,d_X)$ into $\cF$} is a probability distribution $\mu$
on pairs $(\f,(Y,d_Y))$ such that $\f : X \to Y$, $(Y,d_Y)\in \cF$, and
$d_Y(\f(x),\f(x')) \geq d_X(x,x')$ for all $x,x' \in X$.
The {\em expected stretch of $\mu$} is defined by
\ifprocs
\else
\[
  \dist(\mu)
  \eqdef \max \left\{ \frac{\E_{(\f,(Y,d_Y)) \sim \mu} \left[d_Y(\f(x),\f(x'))\right]}{d_X(x,x')} : x\neq x' \in X\right\}.
\]
\fi

We will refer to an undirected graph $G$ equipped with edge lengths $\ell_G : E(G) \to \R_+$
as a {\em metric graph,}
and use $d_G$ to denote the corresponding shortest-path distance.
If $G$ is equipped implicitly with a set $\Term(G) \subseteq V(G)$
of terminals, we refer to it as a {\em terminated graph.}
A graph equipped with both lengths and terminals will be called a {\em terminated metric graph.}
We will consider any graph or metric graph $G$ as terminated
with $\Term(G)=V(G)$ if terminals are not otherwise specified.

Given a terminated metric graph $G$,
a \emph{stochastic terminal embedding of $G$ into a family $\mathcal{F}$ of terminated metric graphs} is a distribution $\mu$ over pairs $(\varphi,F)$
such that $\varphi:V(G)\to V(F)$; the graph $F\in \mathcal{F}$;
the terminals map to terminals:
\[
  \forall t\in \Terminals(G), \qquad
  \Pr \big[ \varphi(t)\in \Terminals(F) \big] = 1\, ;
\]
and the embedding is non-contracting on terminals:
\begin{equation}
  \label{EQ: dominating graph embedding}
  \forall s,t\in \Term(G),\quad
  \Pr_{(\varphi,F)\sim\mu} \big[d_F(\varphi(s),\varphi(t))
    \geq d_G(s,t)\big] = 1\,.
\end{equation}
The \emph{expected stretch} of this embedding, again denoted $\dist(\mu)$,
is defined just as for general metric spaces:
\ifprocs
\else
    \begin{equation}
      \label{EQ: exp distortion graph embedding}
      \dist(\mu) \eqdef \max \left\{ \frac{\E_{(\f,F)\sim \mu}\left[d_F(\f(u),\f(v))\right]}{d_G(u,v)}: u\neq v \in V(G) \right\}.
    \end{equation}
\fi

\begin{theorem}\label{THM: UB embedding into trees}
   Consider a terminated metric plane graph $G$ with $\gamma = \gamma(G;\Term(G))$.
   Then $G$ admits a stochastic terminal embedding into the family of metric trees
   with expected stretch $O(\log \gamma)$.
\end{theorem}

Theorem \ref{THM: UB embedding into trees} immediately yields
Theorem \ref{THM: main result embedding} using the fact that
every finite tree metric embeds isometrically into $\ell_1$
(see, e.g., \cite {GNRS04} for further details).
The bound $O(\log \gamma)$ is optimal up to the hidden constant,
as it is known that for an $m \times m$ planar grid
equipped with uniform edge lengths, the expected stretch
of any stochastic embedding into metric trees is at least $\Omega(\log m)$ \cite{KRS01}.
(A similar lower bound holds for the diamond graphs \cite{GNRS04}.)

Theorem~\ref{THM: UB embedding into trees} may also be of independent interest
(including when $\Term(G)=V(G)$)
as embedding into dominating trees has many applications,
including to competitive algorithms for online problems
such as buy-at-bulk network design \cite{AA97},
and to approximation algorithms for combinatorial optimization, e.g.,
for the group Steiner tree problem \cite{GKR00}.
We remark that stochastic terminal embeddings into metric trees were
employed by~\cite{GNR10} in the context of approximation algorithms,
and were later used in~\cite{EGKRTT14} to design flow sparsifiers.

\section{Approximation by random trees}\label{SEC: preliminaries}

Before introducing our primary technical tools, we will motivate
their introduction with a high-level overview of the proof of Theorem \ref{THM: UB embedding into trees}.
Fix a terminated metric plane graph $G$ with $\gamma=\gamma(G;\Term(G)) > 1$.
Our plan is to approximate $G$ by an OS-instance (where
all terminals lie on a single face) by uniting the $\gamma$
faces covering $\Term(G)$, while approximately preserving the shortest-path metric on $G$.
The use of stochastic embeddings will come from our need to perform this approximation randomly,
preserving distances only in expectation.
Using the known result that OS-instances admit
stochastic terminal embeddings into metric trees, this will complete the proof.

A powerful tool for randomly ``simplifying'' a graph
is the Peeling Lemma~\cite{LS09},
which informally ``peels off'' any subset $A\subset V(G)$ from $G$,
by providing a stochastic embedding of $G$ into graphs obtained by ``gluing''
copies of $G\setminus A$ to the induced graph $G[A]$.
The expected stretch of the embedding
depends on how ``nice'' $A$ is;
for example, it is $O(1)$ when $A$ is a shortest path in a planar $G$.
The Peeling Lemma can be used to stochastically embed $G$ into dominating OS-instances with expected stretch $2^{O(\gamma)}$ \cite[Section 4.5]{CSW13},
by iteratively peeling off a shortest path $A$ between two special faces
(which has the effect of uniting them into a single face).

In contrast, our argument applies the Peeling Lemma only once.
We pick $A$ to form a connected subgraph in $G$ that spans
the $\gamma$ distinguished faces.
By cutting along $A$, one effectively merges all $\gamma$ faces
into a single face in a suitably chosen drawing of $G\setminus A$ .
The Peeling Lemma then provides a stochastic terminal embedding of $G$ into a family
of OS-instances that are constructed from copies of $A$ and $G\setminus A$.

The expected stretch we obtain via the Peeling Lemma
is controlled by how well the (induced) terminated metric graph on $A$
can be stochastically embedded into a distribution over metric trees.
For this purpose, we choose the set $A$ to be a shortest-path tree in $G$
that spans the $\gamma$ distinguished faces, and
then use a result of Sidiropoulos~\cite{Sidiropoulos10}
to stochastically embed $A$ into metric trees
with expected stretch that is logarithmic in the number of {\em leaves}
(rather than logarithmic in the number of vertices, as in
stochastic embeddings for general finite metric spaces \cite{FRT04}).
We remark that this is non-trivial because, while $A$ is (topologically)
a tree spanning $\gamma$ faces, the relevant metric on $A$ is $d_G$
(which is not a path metric on $G[A]$).
 
\subsection{Random partitions, embeddings, and peeling}

\newcommand{\GAa}{G_{\!A}\strut^{\!\!\!a}}
\newcommand{\HAa}{H_{\!A}\strut^{\!\!\!a}}
\newcommand{\GhatA}{\widehat{G}_{\!A}}
\newcommand{\HhatA}{\widehat{H}_{\!A}}

For a finite set $S$, we use $\Trees(S)$ to denote the set of all metric spaces $(S,d)$
that are isometric to $(V(T),d_T)$ for some metric tree $T$.

\begin{theorem}[Theorem 4.4 in \cite{Sidiropoulos10}]\label{THM: log k Tasos}
     Let $G$ be a metric graph, and let $P_1, \ldots,P_m$ be shortest paths in $G$ sharing a common endpoint. Then the metric space $\big(\cup_{i=1}^m V(P_i), d_G \big)$ admits a stochastic embedding into $\Trees(\cup_{i=1}^m V(P_i))$ with expected stretch $O(\log m)$.
\end{theorem}

Let $(X,d)$ be a finite metric space.
A distribution $\nu$ over partitions of $X$ is called $(\beta, \Delta)$-\emph{Lipschitz}
if every partition $P$ in the support of $\nu$ satisfies $S \in P \implies \diam_X(S) \leq \Delta$, and moreover,
\[
  \forall x,y \in X, \qquad
   \Pr_{P\sim \nu}[P(x)\neq P(y)]\leq \beta\cdot \frac{d(x,y)}{\Delta}\,,
\]
where for $x \in X$, we use $P(x)$ to denote the unique set in $P$ containing $x$.

We denote by $\beta_{(X,d)}$ the infimal $\beta\geq 0$ such that for every $\Delta>0$, the metric $(X,d)$ admits a $(\beta, \Delta)$-Lipschitz random partition.
The following theorem is due to Klein, Plotkin, and Rao~\cite{Klein93} and Rao~\cite{Rao99}.

\begin{theorem}\label{THM: beta=1 for planar graphs}
  For every planar graph $G$, we have $\beta_{(V(G),d_{G})} \leq O(1)$.
\end{theorem}

Let $G$ be a metric graph, and consider $A\subseteq V(G)$. The \emph{dilation of $A$ inside $G$} is defined to be
$$\dil_G(A) \eqdef \max_{u,v\in A}\frac{d_{G[A]}(u,v)}{d_G(u,v)}\,,$$
where $d_{G[A]}$ denotes the induced shortest-path distance on the metric graph $G[A]$.

For two metric graphs $G,G'$, a \emph{1-sum of $G$ with $G'$}
is a graph obtained by taking two disjoint copies of $G$ and $G'$,
and identifying a vertex $v\in V(G)$ with a vertex $v'\in V(G')$. This definition naturally extends to a $1$-sum of any number of graphs.
Note that the $1$-sum naturally inherits its length function from $G$ and $G'$.

\subsubsection{Peeling}

Consider a subset $A \subseteq V(G)$.  For $a \in A$, let $\GAa$ denote the graph $G[(V(G)\setminus A) \cup \{a\}]$.
We define the graph $\GhatA$ as the $1$-sum of $G[A]$ with $\{\GAa : a \in A\}$,
where $G[A]$ is glued to each $\GAa$ at their common copy of $a \in A$.
Let us write the vertex set of $\GhatA$ as the disjoint union:
\[
   V(\GhatA) = \hat{A} \sqcup \bigsqcup_{a \in A} \left\{ (a, v) : v \in V(G)\setminus A \right\},
\]
where $\hat{A} \eqdef \left\{ \hat{a} : a \in A \right\}$ represents the canonical image of $G[A]$ in $\GhatA$, and $(a,v)$ corresponds to the image of $v \in V(G) \setminus A$ in $\GAa$.
Say that a mapping $\psi : V(G) \to V(\GhatA)$ is a {\em selector map} if it satisfies:
\begin{enumerate}
   \item For each $a \in A$, $\psi(a)=\hat{a}$.
   \item For each $v \in V(G) \setminus A$, $\psi(v) \in \left\{ (a,v) :a \in A \right\}$.
\end{enumerate}
In other words, a selector maps each $a \in A$ to its unique copy in $\GhatA$, and maps each $v \in V(G) \setminus A$
to one of its $|A|$ copies in $\GhatA$.

\begin{lemma}[The Peeling Lemma~\cite{LS09}]\label{LMA: the peeling lemma}
   Let $G=(V,E)$ be a metric graph and fix a subset $A\subseteq V$.
   Let $G'$ be obtained by removing all the edges inside $A$:
   \[
   G' \eqdef (V, E') \qquad \textrm{with} \qquad E'=E\setminus E(G[A])\,,
   \]
   and denote $\beta=\beta_{(V,d_{G'})}$.
   Then there is a stochastic embedding $\mu$ of $G$ into the metric graph $\GhatA$
   such that $\mu$ is supported on selector maps
   has expected stretch $\dist(\mu) \leq O(\beta \cdot \dil_G(A))$.
\end{lemma}

\begin{remark}\label{RMK: remark peeling lemma}
   The statement of the Peeling Lemma in~\cite{LS09} (see also~\cite{BLS10}) does not specify explicitly all the above details about the selector maps,
   but they can be easily verified by inspecting the proof.
\end{remark}

\subsubsection{Composition}

Consider now some metric tree $T \in \Trees(A)$.
Via the identification between $A$ and $\hat{A} \subseteq V(\GhatA)$, we may consider
the associated metric tree $\hat{T} \in \Trees(\hat{A})$.
Define the metric graph
$\GhatA\llbracket T\rrbracket$
with vertex set $V(\GhatA)$ and edge set
\[
   E(\GhatA\llbracket T\rrbracket) \eqdef \big(E(\GhatA) \setminus E(\GhatA[\hat{A}])\big) \cup E(\hat{T})\,,
\]
where the edge lengths are inherited from $\GhatA$ and $\hat{T}$, respectively.
In other words, we replace the edges of $\GhatA[\hat{A}]$ with those coming from $\hat{T}$.
Finally, denote by
\[
   \cF_{G,A} \eqdef \left\{ \GhatA\llbracket T\rrbracket : T \in \Trees(A) \right\}
\]
the family of all metric graphs arising in this manner.
The following lemma is now immediate.

\begin{lemma}\label{lem:Gbrack}
   Every metric graph in $\cF_{G,A}$ is a $1$-sum of some $T \in \Trees(A)$ with the graphs $\{ \GAa : a \in A \}$.
\end{lemma}

Suppose that $\mu$ is a stochastic embedding of $G$ into $\GhatA$ that is supported on pairs $(\psi, \GhatA)$,
where $\psi$ is a selector map.
Let $\nu$ denote a stochastic embedding of $(A,d_{G})$ into $\Trees(A)$.
By relabeling vertices, we may assume that $\nu$ is supported on pairs $(\id,T)$
where $\id : A \to A$ is the identity map.
Altogether, we obtain a stochastic embedding of $G$ into $\cF_{G,A}$,
which we denote $\nu \circ \mu$ and define by
\ifprocs
\else
\[
  \forall T \in \Trees(A),
  \qquad
  (\nu\circ \mu)(\psi, \GhatA\llbracket T\rrbracket)
  \eqdef \mu(\psi,\GhatA) \cdot \nu(\id, T)\, ,
\]
\fi
where the product between the probability measures $\mu$ and $\nu$
represents drawing from the two distributions independently.
While notationally cumbersome, the following claim is now straightforward.

\begin{lemma}[Composition Lemma]\label{LMA: the composition lemma}
   It holds that
   \[
      \dist(\nu \circ \mu) \leq \dist(\nu) \cdot \dist(\mu)\,.
   \]
\end{lemma}

\subsection{Approximation by OS-instances}
\label{SEC: proof of main result}

Let us now show that every terminated metric plane graph $G$ with $\gamma=\gamma(G;\Term(G))$ admits
a stochastic terminal embedding into OS-instances.
In Section~\ref{SEC: embedding into trees}, we recall how OS-instances
can be stochastically embedded into metric trees, thereby completing the proof
of Theorem \ref{THM: UB embedding into trees}.

Let $F_1,\ldots,F_{\gamma}$ be faces of $G$ that cover $\Term(G)$,
and denote $T_i \eqdef V(F_i) \cap \Term(G)$.
For each $i \geq 1$, fix an arbitrary vertex $v_i \in V(F_i)$.
Denote $r \eqdef v_1$, and for each $i \geq 2$, let $P_i$ be the shortest path from $v_i$ to $r$.
Finally, let $P$ be the tree obtained as the union of these paths,
namely, the induced graph $G[\cup_{i \geq 2} P_i]$.

   We present now Klein's Tree-Cut
   operation \cite{Klein06}. It takes
   as input a plane graph $G$ and a tree $T$ in $G$, and ``cuts open''
   the tree to create a new face $F_{new}$. More specifically, consider walking
   ``around'' the tree and creating a new copy of each vertex and edge of $T$
   encountered along the way. This operation maintains planarity while
   replacing the tree $T$ with a simple cycle $C_T$ that bounds the new
   face. It is easy to verify that $C_T$ has two copies of
   every edge of $T$, and $\deg_T(v)$ copies of every vertex of $T$, where
   $\deg_T(v)$ stands for the degree of $v$ in $T$. This Tree-Cut operation can
   also be found in~\cite{Cora:Thesis,BKK07,BKM09}.

   We apply Klein's Tree-Cut operation to $G$
   and the tree $P$, and let $G_1$ be the resulting metric plane graph
   with the new face $F_{new}$, after we replace $P$ with a simple cycle $C_P$; see
   Figure~\ref{FGR: Klein Tree cut} for illustration. Since $P$ shares at least
   one vertex with each face $F_i$ in $G$ (namely, $v_i$), the
   cycle $C_P$ shares at least one vertex with each face $F_i$ in $G_1$.

   We now
   construct $G_2$ by applying two operations on $G_1$. First, for every face
   $F_i$ that shares exactly one vertex with $C_P$, namely only $v_i$
   (or actually a copy of it), we split this vertex into two
   as follows. Let $N^1_{G_1}(v_i)$ be all the neighbors of
   $v_i$ in $G_1$ embedded between the face $F_i$ and $F_{new}$ on one side,
   and $N^2_{G_1}(v_i)$ be all its neighbors on the other side. We split $v_i$
   into two vertices $v'_i,v''_i$ that are connected by an edge of length $0$,
   and connect all the vertices in $N^1_{G_1}(v_i)$ to $v'_i$ and all the vertices in $N^2_{G_1}(v_i)$ to $v''_i$.
   See Figure~\ref{FGR: add an edge between faces} for illustration.
   Notice that
   this new edge $\{v'_i,v''_i\}$ is incident to both $F_i$ and $F_{new}$, and
   that this operation maintains the planarity, along with the distance metric of
   $G_1$ (in the straightforward sense, where one takes a quotient by vertices at distance $0$
   from each other).

   The second operation adds between all the copies of the same $v\in
   V(P)$ a star with edge length $0$ drawn inside $F_{new}$. Note that adding
   the stars inside $F_{new}$ does not violate the planarity since all the
   copies of the vertices in $C_P$ are ordered by the walk around $P$; see
   Figure~\ref{FGR: Klein Tree cut} for illustration. It is
   easy to verify that if we identify each $v\in V(P)$ with one of its copies
   in $G_2$ arbitrarily then
   \begin{equation}\label{EQ: G and G2 same metric}
       \forall x,y\in V(G), \qquad d_G(x,y)=d_{G_2}(x,y).
   \end{equation}

\ifprocs
\else
 \begin{figure}
    \centering
    \includegraphics[angle=0,width=1\textwidth]{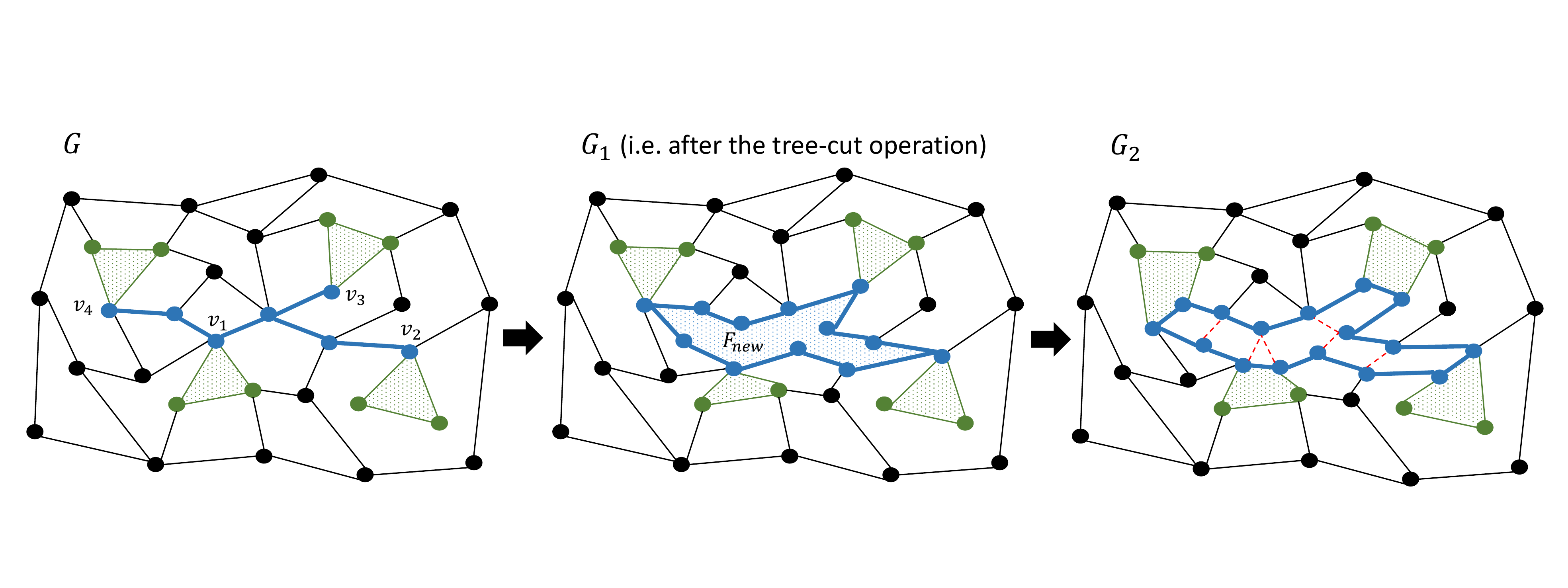}
    \caption{In $G$, the tree $P$ (in blue) is incident to all $\gamma=4$ distinguished special faces (drawn in green).
$G_1$ is obtained by applying the Tree-Cut operation on $G$ and $P$, which creates a new face $F_{new}$.
Finally, $G_2$ is obtained by duplicating some vertices on $F_{new}$
and connecting copies of the same vertex by length-zero edges (the dashed red edges).
}
    \label{FGR: Klein Tree cut}
    \mbox{}\hrule
   \end{figure}

   \begin{figure}
    \centering
    \includegraphics[angle=0,width=0.8\textwidth]{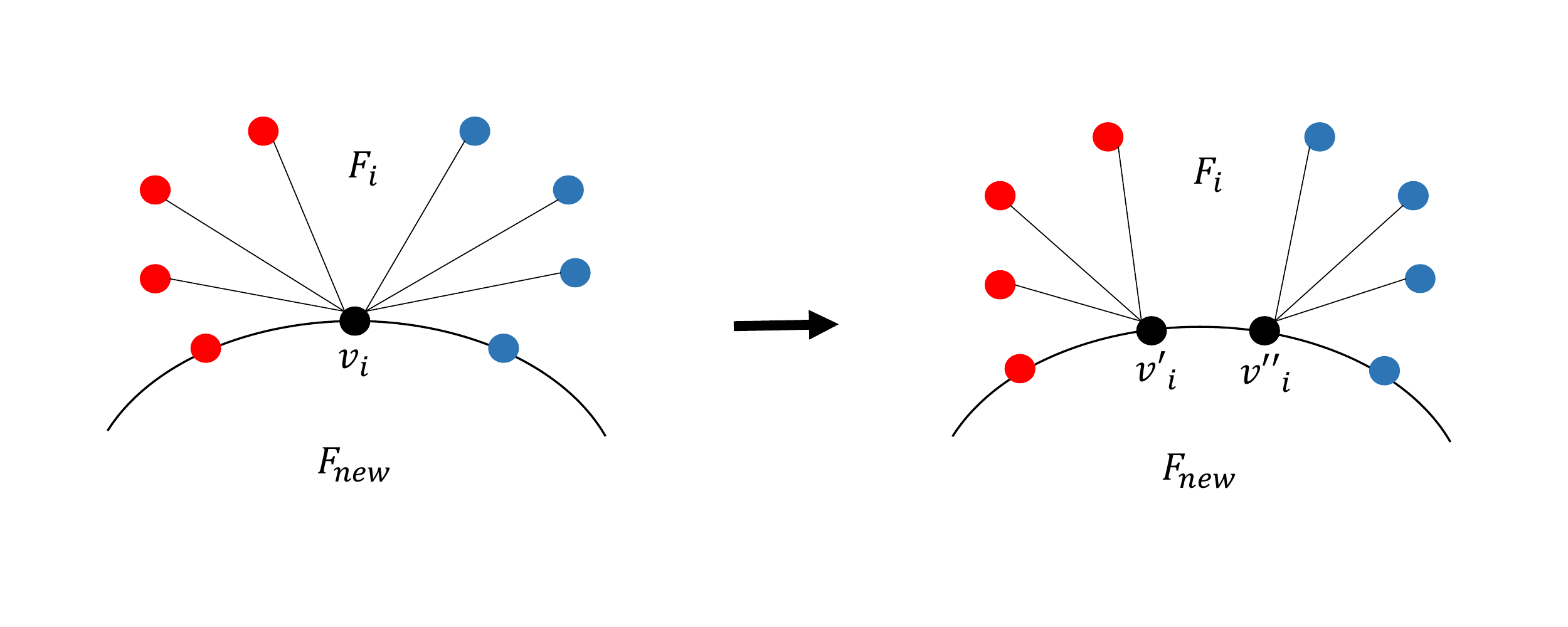}
    \caption{The neighbors of $v_i$ are partitioned into two sets (colored red and blue) by going around $v_i$ in the plane and watching for the location of faces $F_i$ and $F_{new}$, to eventually split $v_i$ into two.
}
    \label{FGR: add an edge between faces}
    \mbox{}\hrule
   \end{figure}

\fi

   \begin{lemma}\label{LMA: embed V(P) to tree distortion log gamma}
      $(V(P),d_{G})$ admits a stochastic embedding into $\Trees(V(P))$ with expected stretch at most $O(\log \gamma)$.
   \end{lemma}

   \begin{proof}
    Apply Theorem~\ref{THM: log k Tasos} on the shortest-paths $P_2,\ldots,P_{\gamma}$ in $G$, with shared vertex $v_1=r$.
   \end{proof}

   Let $A \subseteq V(G_2)$ denote all the vertices on the boundary of $F_{new}$ in $G_2$.
   To every $T \in \Trees(V(P))$, we can associate a tree $T' \in \Trees(A)$ by identifying
   $x \in V(P)$ with one of its copies in $A$, and attaching the rest of its copies to $x$
   with an edge of length $0$.  Using \eqref{EQ: G and G2 same metric} in conjunction with Lemma \ref{LMA: embed V(P) to tree distortion log gamma}
   yields the following.

   \begin{corollary}\label{COR: A tree log gamma}
      $(A, d_{G_2})$ admits a stochastic embedding into $\Trees(A)$ with expected stretch at most $O(\log \gamma)$.
   \end{corollary}

   Let $H$ be the graph obtained from $G_2$ by adding an edge $\{u,v\}$ of length $d_G(u,v)$ between every pair of vertices $u,v\in A$.
   By construction, we have $\dil_{H}(A)=1$. Let $E' \eqdef E(H)\setminus E(H[A])$, and $H'=(V(H),E')$.
   While $H$ is in general non-planar,
   the graph $H'$ and $\HAa$ for $a\in A$ are subgraphs of the planar graph $G_2$, and are thus planar as well, and by Theorem~\ref{THM: beta=1 for planar graphs}
   we have $\beta_{(V(H),d_{H'})} \leq O(1)$.

   By applying the Peeling Lemma (Lemma~\ref{LMA: the peeling lemma}) to $H$ and $A\subseteq V(H)$,
   we obtain a stochastic embedding $\mu$ of $H$ into $\HhatA$ such that $\mu$ is supported on selector maps
   and $\dist(\mu) \leq O(1)$.
    Using Corollary~\ref{COR: A tree log gamma} and the fact that
    $(A,d_{H})$ is the same as $(A,d_{G_2})$, we obtain a stochastic
    embedding $\nu$ of $(A,d_{H})$ into $\Trees(A)$ with $\dist(\nu) \leq
    O(\log \gamma)$.

Define $\Term(H)$ to be the set of vertices in $\Term(G)$ together with all their copies created in the construction of $H$, and
\[
   \Term(\HhatA) \eqdef \{ \hat{a} : a \in \Term(H) \} \cup \{ (v,a) : v \in \Term(H), a\in A \}\,.
\]
By convention, for any subgraph $H'$ of $H$
we have $\Term(H') \eqdef V(H') \cap \Term(H)$.

Applying the Composition Lemma (Lemma~\ref{LMA: the composition lemma})
to the pair $\mu,\nu$ (in conjunction with Lemma~\ref{lem:Gbrack})
yields a stochastic embedding
$\pi \eqdef \nu \circ \mu$ satisfying the next lemma.

\begin{lemma}\label{LMA: G to W distortion log gamma}
      $(V(G),d_G)$ admits a stochastic embedding $\pi$ into the family of metric graphs that are $1$-sums of a metric tree
      with the graphs $\{\HAa : a \in A\}$, where $\HAa$ is glued to $T$ along a vertex of $\Term(\HAa)$,
      and such that $\dist(\pi) \leq O(\log \gamma)$.
      Moreover, every $(\f, W) \in \supp(\pi)$ satisfies $\f(\Term(G)) \subseteq \Term(W)$.
\end{lemma}

It remains to prove that $\pi$ in this lemma is an embedding into OS-instances,
i.e., every $1$-sum in the support of $\pi$ is an OS-instance.
We first show this for every pair $\{ (\HAa, \Term(\HAa)) : a \in A\}$.

\begin{lemma}\label{lem:goodface}
       For every $a\in A$, there is a face $F_a$ in $\HAa$ such that $\Term(\HAa)\subseteq V(F_a)$.
    \end{lemma}

    \begin{proof}
      Fix $a\in A$. The graph $G_2$ is planar, and while $H$ need not be planar,
      the subgraphs $G_2[(V(G_2)\setminus A) \cup \{a\}]$ and $\HAa$ are identical for each $a \in A$.
      Thus, it suffices to prove the lemma for the subgraphs $G_2[(V(G_2) \setminus A)\cup \{a\}]$.

      Observe that if we remove from $G_2$ a vertex $v\in V(G_2)$, then all the
   faces incident to $v$ in $G_2$ become one new face in the graph
$G_2\setminus \{v\}$. Moreover, if we remove from $G_2$ both endpoints of an
edge $\{u,v\}$, then all the faces incident to either $u$ or $v$ become one new
face in $G_2\setminus\{u,v\}$. Recall that $G_2[A]$ is a simple cycle (bounding
$F_{new}$), thus $G_2[A\setminus \{a\}]=G_2[A]\setminus \{a\}$ is connected,
and all the faces incident to at least one vertex in $A\setminus\{a\}$ become
one new face in $G_2[(V(G_2)\setminus A)\cup\{a\}]$, which we denote $F^a_{new}$.

By construction of $G_2$ (which splits a vertex of $G_1$ if
it is the only vertex incident to both $F_i$ and $F_{new}$), every face $F_i$
is incident to at least two vertices in $A$, and thus to at least one in
$A\setminus \{a\}$. It follows that all the terminals in $G_2[(V(G_2)\setminus A)\cup\{a\}]$ are on the same face $F^a_{new}$. In addition, since $a$ has at
least one neighboring vertex $b \in A$, at least one face is incident to both
$a$ and $b$ in $G_2$, and it becomes part of the face $F^a_{new}$ in
$G_2[(V(G_2)\setminus A) \cup \{a\}]$. Therefore, $a \in V(F^a_{new})$ as well,
and the lemma follows.
\end{proof}

\begin{lemma}\label{LMA: every face can be the outerface}
   Suppose $W$ is a planar graph formed from the $1$-sum of a tree $T$
   and a collection of (pairwise disjoint) plane graphs $\{H_a : a \in A\}$, where each $H_a$ has a distinguished face $F_a$,
   and $H_a$ is glued to $T$ along a vertex of $V(F_a)$.
   Then there exists a drawing of $W$ in which all the vertices $V(T) \cup \bigcup_{a \in A} V(F_a)$ lie on the outer face.
\end{lemma}

\begin{proof}
It is well-known that every plane graph can be redrawn so that any desired
face is the outer face (see, e.g., \cite[\S 9]{Whitney32}).
So we may first construct a planar drawing of $T$, and then extend
this to a planar drawing of $W$ where each $H_a$ is drawn so that $F_a$ bounds the image of $H_a$,
and the interior of $F_a$ contains only the images of vertices in $V(H_a)$.
\end{proof}

Combining Lemmas~\ref{LMA: G to W distortion log gamma},
\ref{lem:goodface} and~\ref{LMA: every face can be the outerface}
yields the following corollary.

\begin{corollary}\label{COR: terminals outerface in W}
  $G$ admits a stochastic embedding with expected stretch $O(\log \gamma)$
  into a family $\cF$ of terminated metric plane graphs,
  where each $W \in \cF$ satisfies $\gamma(W;\Term(W))=1$.
\end{corollary}

Note that in the stochastic embedding of this corollary,
the stretch guarantee applies to all vertices (and not only to terminals),
and the choice of terminals restricts the host graphs $W\in\cF$,
as they are OS-instances.

\subsection{From OS-instances to random trees}
\label{SEC: embedding into trees}

We need a couple of known embedding theorems.

\begin{theorem}[{\cite[Thm. 5.4]{GNRS04}}]
   \label{thm:gnrs}
   Every metric outerplanar graph admits a stochastic embedding into metric trees with expected stretch $O(1)$.
\end{theorem}

The next result is proved in~\cite[Thm. 4.4]{LMM15}
(which is essentially a restatement of \cite[Thm. 12]{EGKRTT14}).

\begin{theorem}
   \label{thm:retract}
   If $G$ is a terminated metric plane graph and $\gamma(G;\Term(G))=1$, then
   $G$ admits a stochastic terminal embedding into metric outerplanar graphs with expected stretch $O(1)$.
\end{theorem}

In conjunction with Theorem~\ref{thm:gnrs}, this shows that every OS-instance admits
a stochastic terminal embedding into metric trees with expected stretch $O(1)$.
Combined with Corollary~\ref{COR: terminals outerface in W}, this finishes
the proof of Theorem~\ref{THM: UB embedding into trees}.

\section{Polymatroid flow-cut gaps}\label{SEC: vertex capacitated flow cut gap}

\newcommand{\e}{\epsilon}
\newcommand{\seteq}{\eqdef}
\newcommand{\len}{\ell}

We now discuss a network model introduced in \cite{CKRV15} that
generalizes edge and vertex capacities.
Recall that if $S$ is a finite set, then a function $f : 2^S \to \R$ is called {\em submodular}
if $f(A)+f(B) \geq f(A \cap B) + f(A \cup B)$ for all subsets $A,B \subseteq S$.
For an undirected graph $G=(V,E)$, we let $E(v)$ denote the set of edges incident to $v$.
A collection $\vec{\rho} = \{ \rho_v : 2^{E(v)} \to \R_+ \}_{v \in V}$ of monotone, submodular
functions are called {\em polymatroid capacities on $G$.}

Say that a function $\varphi : E \to \R_+$ is {\em feasible with respect to $\vec{\rho}$}
if it holds that for every $v \in V$ and subset $S \subseteq E(v)$,
it holds that $\sum_{e \in S} \f(e) \leq \rho_v(S)$.
Given demands $\dem : V \times V \to \R_+$, one defines the {\em maximum concurrent flow value}
of the polymatroid network $(G,\vec{\rho},\dem)$,
denoted $\mcf_G(\vec{\rho},\dem)$,
as the maximum value $\e > 0$
such that one can route an $\e$-fraction of all demands simultaneously
using a flow that is feasible with respect to $\vec{\rho}$.

For every subset $S \subseteq E$, define the cut semimetric $\sigma_S : V \times V \to \{0,1\}$
by $\sigma_S(u,v) \seteq 0$ if and only if there is a path from $u$ to $v$ in the graph $G(V,E \setminus S)$.
Say that a map $g : S \to V$ is {\em valid} if it maps every edge in $S$ to one of its two endpoints in $V$.
One then defines the {\em capacity of a set $S \subseteq E$} by
\[
   \nu_{\vec{\rho}}(S) \seteq \min_{\substack{g : S \to V \\ \mathrm{valid}}} \sum_{v \in V} \rho_v(g^{-1}(v))\,.
\]
The {\em sparsity of $S$} is given by
\[
   \Phi_G(S; \vec{\rho},\dem) \seteq \frac{\nu_{\vec{\rho}}(S)}{\sum_{u,v \in V} \dem(u,v) \sigma_S(u,v)}\,.
\]
We also define $\Phi_G(\vec{\rho},\dem) \seteq \min_{\emptyset \neq S \subseteq V} \Phi(S;\vec{\rho},\dem).$
Our goal in this section is to prove the following theorem.

\begin{theorem}\label{thm:poly-flowcut}
   There is a constant $C \geq 1$ such that the following holds.
   Suppose that $G=(V,E)$ is a planar graph and $D \subseteq F_1 \cup F_2 \cup \cdots \cup F_{\gamma}$,
   where each $F_i$ is a face of $G$.
   Then for every collection $\vec{\rho}$ of polymatroid capacities on $G$ and every set of demands $\dem : D \times D \to \R_+$
   supported on $D$, it holds that
   \[
      \mcf_G(\vec{\rho},\dem) \leq \Phi_G(\vec{\rho},\dem) \leq C \gamma \cdot \mcf_G(\vec{\rho},\dem)\,.
   \]
\end{theorem}

\subsection{Embeddings into thin trees}

In order to prove this, we need two results from \cite{LMM15}.
Suppose $G$ is an undirected graph, $T$ is a connected tree,
and $f : V(G) \to V(T)$.
For every distinct pair $u,v \in V(G)$, let $P^T_{uv}$ denote the unique simple
path from $f(u)$ to $f(v)$ in $T$.
Say that the map $f$ is {\em $\Delta$-thin} if, for every $u \in V(G)$,
the induced subgraph on $\bigcup_{v : \{u,v\} \in E(G)} P^T_{uv}$ can be covered
by $\Delta$ simple paths in $T$ emanating from $f(u)$.

Suppose further that $G$ is equipped with edge lengths $\len : E(G) \to \R_+$.
If $(X,d_X)$ is a metric space and $f : V(G) \to X$, we make
the following definition.
For $\tau > 0$ and any $u \in V(G)$:
\ifprocs
\else
    \[
       |\nabla_{\tau} f(u)|_{\infty} \seteq \max \left\{ \frac{d_X(f(u),f(v))}{\len(u,v)} : \{u,v\} \in E \textrm{ and } \len(u,v) \in [\tau,2\tau] \right\} .
    \]
\fi

\begin{fact}\label{fact:linemap}
   Suppose that $f : V(G) \to \R$ is $1$-Lipschitz, where $V(G)$ is equipped with the path metric $d_{G,\ell}$.
   Then $f$ is $2$-thin and
   \[\max \left\{ |\nabla_{\tau} f(u)|_{\infty} : u \in V(G), \tau > 0 \right\}\leq 1\,.\]
\end{fact}

\begin{theorem}[Rounding theorem~\cite{LMM15}]
   \label{thm:rounding}
   Consider a graph $G=(V,E)$ and a subset $D \subseteq V$.
   Suppose that for every length $\len : E \to \R_+$, there is a random $\Delta$-thin mapping
   $\Psi : V \to V(T)$ into some random tree $T$ that satisfies:
   \begin{enumerate}
      \item For every $v \in V$ and $\tau > 0$:  $\E |\nabla_{\tau} \Psi(v)|_{\infty} \leq L$.
      \item For every $u,v \in D$:
         \[
            \E\left[d_T(\Psi(u),\Psi(v))\right] \geq \frac{d_{G,\len}(u,v)}{K}\,.
         \]
   \end{enumerate}
   Then for every collection $\vec{\rho}$ of polymatroid capacities on $G$ and every set of demands $\dem : D \times D \to \R_+$
   supported on $D$, it holds that
   \[
      \mcf_G(\vec{\rho},\dem) \leq \Phi_G(\vec{\rho},\dem) \leq O(\Delta KL) \cdot \mcf_G(\vec{\rho},\dem)\,.
   \]
\end{theorem}

\begin{theorem}[Face embedding theorem~\cite{LMM15}]
   \label{thm:face-embed}
      Suppose that $G=(V,E)$ is a planar graph and $D \subseteq V$ is a subset
      of vertices contained in a single face of $G$.  Then for every $\len : E \to \R_+$,
      there is a random $4$-thin mapping $\Psi : V \to V(T)$ into a random tree metric
      that satisfies the assumptions of Theorem~\ref{thm:rounding} with $K,L \leq O(1)$.
\end{theorem}

We now use this to prove the following multi-face embedding theorem;
combined with Theorem~\ref{thm:rounding}, it yields Theorem~\ref{thm:poly-flowcut}.

\begin{theorem}[Multi-face embedding theorem]
   Suppose that $G=(V,E)$ is a planar graph and $D \subseteq F_1 \cup F_2 \cup \cdots \cup F_{\gamma}$,
   where each $F_i$ is a face of $G$.
   Then for every $\len : E \to \R_+$, there is a random $4$-thin mapping $\Psi : V \to V(T)$
   into a random tree metric that satisfies the assumptions of Theorem~\ref{thm:rounding} with $L \leq O(1)$ and $K \leq O(\gamma)$.
\end{theorem}

\begin{proof}
   For each $i=1,2,\ldots,\gamma$, let $\Psi_i : V \to V(T_i)$ be the random $4$-thin mapping
   guaranteed by Theorem~\ref{thm:face-embed} with constants $1 \leq K_0,L_0 \leq O(1)$,
   and let $\Psi'_i : V \to \R$ be the $2$-thin mapping given by $\Psi'_i(v) = d_{G,\len}(v, F_i)$
   (recall Fact~\ref{fact:linemap}).
   Now let $\Psi : V \to V(T)$ be the random map that arises from choosing one of $\{\Psi_1,\ldots,\Psi_{\gamma},\Psi'_1,\ldots,\Psi'_{\gamma}\}$
   uniformly at random.
   Then $\Psi$ is a random $4$-thin mapping satisfying (1) in Theorem~\ref{thm:rounding} for some $L \leq O(1)$.

   \medskip

   Consider now some $u \in F_i$ and $v \in V$.
   Let $u' \in F_i$ be such that $d_{G,\len}(v,u')=d_{G,\len}(v,F_i)$.
   If $d_{G,\len}(u',v) \geq \frac{d_{G,\len}(u,v)}{4K_0 L_0}$, then
   \ifprocs
\else
           \[
           \E\left[d_T(\Psi(u),\Psi(v))\right] \geq \frac{1}{2 \gamma} \left|\Psi'_i(u) - \Psi'_i(v)\right| = \frac{d_{G,\len}(u',v)}{2\gamma}
           \geq \frac{d_{G,\len}(u,v)}{8\gamma K_0 L_0}\,.
           \]
   \fi
   If, on the other hand, $d_{G,\len}(u',v) < \frac{d_{G,\len}(u,v)}{4 K_0 L_0}$, then
   \ifprocs
\else
       \begin{align*}
          \E\left[d_T(\Psi(u),\Psi(v))\right] &\geq \frac{1}{2 \gamma} \E\left[d_{T_i}(\Psi_i(u), \Psi_i(v))\right]  \\
                                              &\geq \frac{1}{2 \gamma} \E\left[d_{T_i}(\Psi_i(u), \Psi_i(u'))-d_{T_i}(\Psi_i(u'),\Psi_i(v))\right]  \\
                                              &\geq \frac{1}{2 \gamma} \left(\frac{d_{G,\len}(u, u')}{K_0} -L_0 \,d_{G,\len}(u',v)\right)  \\
                                              &\geq \frac{1}{2\gamma} \left(\frac{d_{G,\len}(u,v)-d_{G,\len}(u',v)}{K_0} - \frac{d_{G,\len}(u,v)}{4K_0}\right) \\
                                              &\geq \frac{1}{2\gamma} \left(\frac34 \frac{d_{G,\len}(u,v)}{K_0} - \frac{d_{G,\len}(u',v)}{K_0}\right) \\
                                              &\geq \frac{d_{G,\len}(u,v)}{4\gamma K_0}\,.
       \end{align*}
   \fi
   Thus $\Psi$ also satisfies (2) in Theorem~\ref{thm:rounding} with $K \leq O(\gamma)$, completing the proof.
\end{proof}

\subsubsection*{Acknowledgments}
We thank Anupam Gupta for useful discussions at an early stage of the research.


\newcommand{\etalchar}[1]{$^{#1}$}

\end{document}